\documentclass{article}
\usepackage[utf8]{inputenc}
\usepackage[noadjust]{cite}
\usepackage{hyperref}
\usepackage{enumerate}
\usepackage{amsmath}
\usepackage{amssymb}
\usepackage{amsthm}
\usepackage{authblk}
\usepackage{bm}
\usepackage{comment}
\usepackage[linesnumbered,ruled,noline]{algorithm2e}

\usepackage{tikz}
\usetikzlibrary{patterns, decorations.pathreplacing}
\usepackage{subcaption}

\newcommand{\hdist}{\texttt{heavy\_dist}}
\newcommand{\dist}{\texttt{dist}}

\newtheorem{definition}{Definition}
\newtheorem{theorem}{Theorem}

\newtheorem{lemma}[theorem]{Lemma}

\newcommand{\Oish}{\widetilde{O}}

\usepackage[colorinlistoftodos,textsize=small,color=red!25!white,obeyFinal]{todonotes}

\usepackage{placeins}
\title{Fast Construction of 4-Additive Spanners}
\author{Bandar Al-Dhalaan}
\affil{bandar@umich.edu\\ Computer Science and Engineering\\ University of Michigan}
\date{}

\begin{document}

\maketitle
\thispagestyle{empty}

\begin{abstract}
%Attach uniqueness of result for 4-additive spanners to second sentence

    A $k$-additive spanner of a graph is a subgraph that preserves the distance between any two nodes up to a total additive error of $+k$. Efficient algorithms have been devised for constructing 2 [Aingworth et al. SIAM '99], 6 [Baswana et al. ACM '10, Woodruff ICALP '13], and 8-additive spanners [Knudsen '17], but no efficient algorithms for 4-additive spanners have yet been discovered. In this paper we present a modification of Chechik's 4-additive spanner construction [Chechik SODA '13] that produces a 4-additive spanner on $\Oish(n^{7/5})$ edges, with an improved runtime of $\Oish(mn^{3/5})$ from $O(mn)$. We also discuss generalizations to the setting of weighted additive spanners.

\end{abstract}
\clearpage
\pagenumbering{arabic}

\section{Introduction}

    \begin{table}[t]
    \begin{center}
    \begin{tabular}{|c|c|c| }
    \hline
    Additive Stretch & Size & Time \\
    \hline \hline
         +2 &  $O(n^{3/2})$ & $O(mn^{1/2})$ \cite{aingworth1999fast}  \\
         \hline
         +4 & $\Oish(n^{7/5})$ & $O(mn)$  \cite{chechik2013new} \\
         \hline
         \textbf{+4} & $\bm{\Oish(n^{7/5})}$ & $\bm{\Oish(mn^{3/5})}$ \textbf{(this paper)}\\
         \hline
         +6 & $\Oish(n^{4/3})$ &  $\Oish(n^2)$ (\cite{Baswana2010} and \cite{woodruff2010additive})\\
         \hline
         +8 & $O(n^{4/3})$ & $O(n^2)$ \cite{knudsen2017additive}\\
         \hline
    \end{tabular}
    \caption{Notable Additive Spanner Constructions}
    \end{center}
    \end{table}
%\gbnoteinline{Note on Table 1: Chechik doesn't get this runtime -- worth separating out from the table to emphasize that your result is new}

A graph on $n$ nodes can have on the order of $m = O(n^2)$ edges. For very large values of $n$, this amount of edges can be prohibitively expensive, both to store in space and to run graph algorithms on. Thus it may be prudent to operate instead on a smaller approximation of the graph. A \textit{spanner} is a type of subgraph which preserves distances between nodes up to some error, which we call the stretch. Spanners were introduced in \cite{peleg1989graph} and \textit{additive} spanners were first studied in \cite{liestman1993additive}.
%\gbnote{Cite original spanners papers here: "Introduced in cite\{PU89\}, additive error first studied in cite\{LS93\}"}

\begin{definition}[Additive Spanners]

A $k-$additive spanner (or ``$+k$ spanner") of a graph $G$ is a subgraph $H$ that satisfies $\dist_H(s,t) \leq \dist_G(s,t) + k$ for each pair of nodes $s,t \in V(G)$. $k$ is called the (additive) \textit{stretch} of the spanner.

\end{definition}

Note that since $H$ is a subgraph of $G$, the lower bound $\dist_G(s,t) \leq \dist_H(s,t)$ is immediate (the error is one-sided). Spanners have found applications in distance oracles \cite{baswana2006faster}, parallel and distributed algorithms for computing almost shortest paths \cite{cohen1998fast,elkin2005approximating}, synchronizers \cite{peleg1987optimal}, and more. 

Spanners are a tradeoff between the size of the subgraph and the stretch; spanner size can be decreased at the cost of a greater stretch and vice versa. The +2 spanner construction due to Aingworth et al. produces spanners of size $O(n^{3/2})$, and this size is optimal \cite{aingworth1999fast}. The +4 spanner construction is due to Chechik, which produces smaller spanners of size $\Oish(n^{7/5})$ (though this bound is not known to be tight; it is conceivable that further improvements may reduce size up to $O(n^{4/3})$)\cite{chechik2013new}. The +6 construction due to Baswana et al. \cite{Baswana2010} and +8 construction due to Knudsen \cite{knudsen2017additive} achieve an $O(n^{4/3})$ spanner. It is known that any $+k$ spanner construction has a lower bound of $n^{4/3 - o(1)}$ edges on the spanners it produces, so error values greater than +6 are not existentially of interest; the +8 spanner exchanges error for a polylog improvement in construction speed. 
%\gbnoteinline{At some point (probably here), you should mention all the size/stretch tradeoffs for additive spanners that are known regardless of construction time, including the lower bound that a $+c$ spanner construction needs $\ge n^{4/3 - o(1)}$ edges (to explain why you don't discuss any more than these four stretch values).}
%\gbnoteinline{Related: worth emphasizing that the $+8$ construction in the table is not interesting existentially, rather it is paying unneeded error in exchange for a polylog improvement in construction time.  The message ``people are willing to pay error for polylog improvements, but here I am getting polynomial improvements at no cost in error'' is part of what makes your paper cool.}
%\gbnoteinline{Also worth citing Knudsen ``Additive Spanners: A Simple Construction" in this discussion.  It doesn't speed things up, but it is an important paper that works purely on ``better'' construction algorithms rather than size/stretch tradeoff.}

In addition to finding spanner constructions that produce the smallest spanners possible, it is also in our interest that these construction algorithms be fast. Because spanners are meant to make graphs more compact, they are mainly of interest for very large graphs. Thus, for very large $n$, a polynomial time speedup to an algorithm for producing $+k$ spanners is highly desirable. There is a long line of work done in the interest of speeding up spanner constructions, including \cite{ roditty2004dynamic, baswana2007simple,woodruff2010additive,knudsen2017additive, alstrup2019constructing}. %A related topic of interest is improving the efficiency of \textit{multiplicative} spanner construction, such as in \cite{rodi}. 
For a comprehensive survey, see \cite{ahmed2020graph}. 
% https://arxiv.org/abs/1709.01960
% survey: https://arxiv.org/abs/1909.03152

Some additive spanner size and efficiency results are summarized in Table 1 above. As mentioned above, the $+8$ spanner due to Knudsen \cite{knudsen2017additive} exchanges error over the $+6$ spanner for a polylog improvement in construction time. In this paper, we present a polynomial speedup to Shiri Chechik's 4-additive spanner construction presented in \cite{chechik2013new}, at no cost in error.
%\gbnote{Format: cite\{a, b, c\} not cite\{a\}, cite \{b\}, cite \{c\}} 

%add table of best results for other spanner runtimes/edges

% -define spanners

% -mention/cite applications (<= 2 sentences)

% -1-2 paragraphs on why you should compute quickly

% -cite papers on fast spanner computation with this paragraph, say "there is a long line of work..."
%     -Baswana-Sen (multiplicative spanners)
%     -Halperin-Zwick (folklore construction, mentioned in some other paper, multiplicative spanners)
%     -Roditty-Zwick (multiplicative spanners)
%     -Woodruff (+6 spanner)
%     -Knudsen (+8 spanner)
%     -others?
%\gbnoteinline{Re-Emphasizing: the citations for multiplicative spanner constructions would be good to include.  Let me know if you have trouble finding these papers}

% -focus in this paper is on Chechik's +4 spanner. we show:

\begin{theorem} [Main Result]
There is an algorithm that constructs (with high probability) a $+4$ spanner on $\Oish(n^{7/5})$ edges in $\Oish(mn^{3/5})$ time.
\end{theorem}

For comparison, the bottleneck to Chechik's original construction is solving the All-Pairs-Shortest-Paths (APSP) problem; with combinatorial methods, this has an $O(mn)$ runtime, and with matrix multiplication methods, $O(n^\omega)$ \cite{seidel1995all} \footnote{We note that when $m> n^{\omega-0.4}$, the bottleneck in the algebraic case is instead the second stage of the algorithm described in section 1.2 ($\Oish(mn^{2/5})$).}. Currently, $\omega <2.373$ \cite{alman2021refined}. See Section 1.2 for a full runtime analysis of the original construction. %\gnote{Cite Seidel's alg}

Our speedup relies on avoiding the APSP problem. We do this by realizing that we can weaken the path finding methods in Chechik's original construction without compromising error. In particular, we introduce a new problem in Section 2.2, which we call the ``Weak Constrained Single Source Shortest Paths" (weak CSSSP) problem. We give a Dijkstra-time solution to this problem and apply it to create our new +4 spanner construction in Section 2.3.

%\gbnote{Not the right citation here -- Seidel's algorithm proves APSP in $n^{\omega}$ time.  This citation gives the current bound on $\omega$, which you can also mention ("currently, $\omega < 2.373...$ [cite alman]")}

%state of the art matrix mult: https://arxiv.org/abs/2010.05846

With matrix multiplication methods for APSP, Chechik's algorithm has an $\Oish(n^\omega)$ implementation \cite{seidel1995all}.%\gbnote{Cite this sentence to paper "On the All-Pairs-Shortest-Path Problem in Unweighted Undirected Graphs" by Seidel.  Also I think the bound is $\Oish(n^{\omega})$ if I remember right.}
However, we note that there is a range of $m$ values where our combinatorial algorithm potentially outperforms algebraic methods; for example if $\omega \geq 2.3$, then when $m < n^{1.699}$, the complexity of our algorithm is polynomially faster.%\gbnote{Only true relative to current value of $\omega$, but not if $\omega=2$, right?  Worth saying that if so (then ``for example, if $\omega \ge 2.3 \dots$'')}

Finally, we extend our weak CSSSP method to the construction of spanners in the weighted setting. While the unweighted setting is predominant in the study of additive spanners, the generalization was made to the weighted setting by Elkin et al. in \cite{elkin2019almost}, and further work was done in \cite{ahmed2020weighted,elkin2020improved,ahmed2021additive}. There are two different weighted generalizations of $+k$ spanners in the literature; the weaker generalization (introduced in \cite{ahmed2020weighted}) requires the spanner $H$ to satisfy $\dist_H(s,t) \leq \dist_G(s,t) + kW$ for each $s,t \in V(G)$, where $W = \max_{e \in E(G)} w(e)$ is the maximum edge weight of the edges of $G$. These are called $+kW$ spanners. The stronger generalization (studied in \cite{elkin2020improved, ahmed2021additive}) defined below restricts the edge weight stretch factor to the maximal edge weight over shortest $s \leadsto t$ shortest paths, denoted $W(s,t)$. For this reason this generalization is known as ``local weighted error".

\begin{definition}[Weighted Additive Spanner (Global Error)]
A $+kW$ spanner of a graph $G$ is a subgraph $H$ that satisfies $\dist_H(s,t) \leq \dist_G(s,t)+kW$ for all $s,t \in V(G)$, where $W = \max_{e \in E}w(e)$
\end{definition}

\begin{definition}[Weighted Additive Spanner (Local Error)]
A $+kW(\cdot,\cdot)$ spanner of a graph $G$ is a subgraph $H$ that satisfies $\dist_H(s,t) \leq \dist_G(s,t)+kW(s,t)$ for all $s,t \in V(G)$, where $W(s,t)$ is the maximum edge weight along a shortest path $\pi_G(s,t)$ in $G$.
\end{definition}

In \cite{ahmed2021additive}, Ahmed et al. generalized the 4-additive spanner construction presented in \cite{chechik2013new} to the (strong) weighted setting. Their algorithm constructs a $+4W(\cdot,\cdot)$ spanner on $\Oish(n^{7/5})$ edges, and can be implemented in $\Oish(mn^{4/5})$ time by using AliAbdi et al's Bi-SPP algorithm \cite{AliAbdi2019} for constrained shortest path finding. Our contribution in section 3 will be a $+4W(\cdot,\cdot) + \epsilon W$ spanner in $\Oish(mn^{3/5})$ time, on $\Oish_\epsilon(n^{7/5} \epsilon^{-1})$ edges with high probability, for any error $1>\epsilon >0$.

%Is this good as an existential statement, or would you change anything? 
%Greg: I'd suggest some changes, see email
\begin{theorem}
For any weighted graph $G = (V,E)$ and $\epsilon >0$, there is a $+4W(\cdot,\cdot)+\epsilon W$ spanner on $\Oish_\epsilon(n^{7/5})$ edges and computable in $\Oish(mn^{3/5})$ time, with high probability.
\end{theorem}
%\gbnoteinline{Once stabilized, add discussion of weighted results in here too}
\noindent
We note that while Ahmed et. al's construction in  \cite{ahmed2021additive} doesn't have the extra $+\epsilon W$ stretch, our construction comes with a polynomial speedup. 

\subsection{Notations}

We will use $\pi_G(s,t)$ to refer to a canonical shortest path between two nodes $s,t \in V(G)$. $P(s,t)$ is a variable we use in some of our algorithms that describes some computed $s \leadsto t$ path. For a node $v \in V(G)$, $\Gamma_G(v)$ denotes the neighborhood of $v$ (the set containing $v$ and its neighbors) in $G$. When $S$ is a set, $\Gamma_G(S):= \bigcup_{v \in S} \Gamma_G(v)$. $\mathcal{P}(u,v)$ denotes the set of all paths between nodes $u,v$. If $A,B$ are subsets of $V$, then $\mathcal{P}(A,B) = \bigcup_{u\in A, v\in B}\mathcal{P}(u,v)$

\subsection{Current Runtime of the +4 Spanner Construction}

%{\color{red} say that Chechik didn't actually analyze her runtime}

In \cite{chechik2013new}, Shiri Chechik presents a spanner construction that produces a +4 spanner on $\Oish(n^{7/5})$ edges on average with probability $>1-1/n$. The runtime complexity, however, was not analyzed. In this section, we will describe Chechik's algorithm and then give a runtime analysis. Chechik's construction of a +4 spanner $H$ of an input graph $G$ can be split into three stages:

\begin{enumerate}[(i)]
    \item All edges adjacent to ``light nodes" (nodes with degree $< \mu = \lceil n^{2/5}\log^{1/5}n\rceil$) are added to $H$.
    
    \item Nodes are sampled for inclusion into a set $S_1$ with probability $9\mu/n$. BFS (Breadth-First Search) trees for these nodes are computed, and the edges for these trees are added to $H$.
    
    \item Nodes are sampled for inclusion into a set $S_2$ with probability $1/\mu$. For each ``heavy" node (nodes with degree $\geq \mu$ in the original graph) $v$ that is not in $S_2$, but is adjacent to some node of $S_2$, we arbitrarily choose a neighbor $x \in S_2$ and add the edge $(v,x)$ to $H$. These choices also define the ``clusters" of the graph: for each $x \in S_2$, $C(x)$ is the set containing $x$ and its adjacent heavy nodes that were paired with $x$ in the previous step. We now find, for each pair $x_1,x_2 \in S_2$, the shortest path $P(s,t)$ subject to the constraint that $s \in C(x_1)$, $t \in C(x_2)$, and $P(s,t)$ has $\leq \mu^3/n$ heavy nodes. We use $\hdist_G(P(s,t))$ to refer to the number of heavy nodes on $P(s,t)$ in $G$.

\end{enumerate}

Algorithm \ref{alg:chechik} gives the full details. Stage (i) takes $O(n)$ time and stage (ii) takes $\Oish(mn^{2/5})$ time with high probability. The computationally dominant step of this algorithm is the task of finding these shortest paths between the clusters in (iii) (unless algebraic methods are used, in which case stage (ii) dominates). For worst case inputs, the expected number of clustered nodes (nodes in some cluster) is $\Omega(n)$. Thus, this algorithm's runtime will be bottlenecked by the all-pairs-shortest-paths problem. We now show that the heavy-node constraint on the paths does not increase the runtime.
To see this, we note that it's enough to search over paths of the form
$$P(x_1, x_2) = (x_1,s)\circ \pi_G(s,t) \circ (t,x_2) \footnote{Paths of this form will be important in our own construction (see Definition 5).}$$
($\circ$ denotes path concatenation), where $x_1,x_2$ range over $S_2$ and $s \in C(x_1)$, $t \in  C(x_2)$. Specifically, we want the shortest path of this form for each pair of clusters $C(x_1),C(x_2)$, where $\pi_G(s,t)$ is a constraint-satisfying ($\leq\mu^3/n$ heavy nodes) path that is also a shortest path in $G$.

To find such paths, we first solve APSP ($O(mn)$ time combinatorially, $O(n^{\omega})$ time algebraically) to get shortest paths $\pi_G(s,t)$ for each $s,t \in V$. Then for all pairs of clustered nodes $s,t$, with cluster centers $x_1,x_2$ respectively: if $\hdist_G(\pi_G(s,t)) \leq \mu^3/n$, set $P(x_1,x_2) = (x_1,s) \circ \pi_G(s,t) \circ (t,x_2)$ as the current best path for the cluster pair $(C(x_1),C(x_2))$ if one hasn't yet been selected, otherwise replace the current path iff $P(x_1,x_2)$ is shorter. This is an APSP time process for finding the shortest valid canonical shortest path connecting each cluster pair; at the end of the process, we add the edges of these best paths. We note that because we're not searching over \textit{all} paths, but only one set of canonical shortest paths, it's possible we fail to find valid (constraint satisfying) paths between some cluster pairs. This does not impede correctness, as we only require these paths in the cases that they exist. 
%\gbnote{Probably worth 1 or 2 sentences mentioning runtime of other, non-dominant steps}

\begin{figure}[h]
\begin{algorithm}[H]
    \label{alg:chechik}\caption{Chechik's 4-Additive Spanner Construction \cite{chechik2013new}}
    \KwIn{$n$-node graph $G=(V,E)$}\vspace{1em}
    $E'=$ All edges incident to light nodes\\
    Sample a set of nodes $S_1$ at random, every node with probability $9\mu/n$\\
    \ForEach{node $x\in S_1$}{
        Construct a BFS tree $T(x)$ rooted at $x$ spanning all vertices in $V$\\
        $E'=E'\cup E(T(x))$\\
    }\vspace{1em}
    Sample a set of nodes $S_2$ at random, every node with probability $1/\mu$\\
    \ForEach{heavy node $x$ so that $(\{x\}\cup\Gamma_G(x))\cap S_2 = \varnothing$}{
        Add all incident edges of $x$ to $E'$\\
    }
    \ForEach{node $x\in S_2$}{
        C(x) = \{x\}
    }
    \ForEach{heavy node $v$ so that $v\not\in S_2$ and $\Gamma(v)\cap S_2\neq\varnothing$}{
        Arbitrarily choose one node $x\in\Gamma_G(v)\cap S_2$\\
        $C(x) = C(x) \cup \{x\}$\\
        $E' = E' \cup \{(u,v)\}$\\
    }
    \ForEach{pair of nodes $(x_1, x_2)$ in $S_2$}{
        Let $\hat{\mathcal{P}}=\{P\in\mathcal{P}(C(x_1),C(x_2))\ |\ \hdist_G(P)\leq\mu^3/n\}$\\
        Let $P(\hat y_1, \hat y_2)$ be the path in $\hat{\mathcal{P}}$ with minimal $\left|P(\hat y_1, \hat y_2)\right|$
        $E'=E'\cup E(P(\hat y_1, \hat y_2))$
    }
    \Return $H=(V,E')$
\end{algorithm}
\end{figure}

\FloatBarrier 

% -state original, mention Chechik doesn't analyze runtime
% -sketch ``APSP-Time" algorithm: requires $\Oish(mn)$ time (``combinatorially") or $\Oish(n^{\omega})$ time (``Seidel's Algorithm", algebraically)
%     - some range of $m$ where the new runtime improves on things!
%     - also good to have fast combinatorial algorithms
%     - (weighted setting --> only $\Oish(mn)$ time)
% - mention other paper you found, which gets $\Oish(mn^{4/5})$ (make clear that this was not known)

\section{Fast Construction of the +4 Spanner}

In this section, we present our main result; a modification of Chechik's +4 spanner construction that has $\Oish(mn^{3/5})$ runtime with high probability, with no compromise to size or error.

\subsection{Constrained Shortest Paths}

Chechik's original algorithm required the computation of shortest paths subject to a constraint on the number of heavy nodes in the paths. This was a proxy for constraining the number of edges that had not yet been added to the spanner at that point in the construction. We will call CSSSP the ``Constrained Single Source Shortest Path Problem". This is similar to the GB-SPP (``Gray-Vertices Bounded Shortest Path Problem") presented by AliAbdi et al. in \cite{AliAbdi2019}, but our constraint is on the edges instead of on the nodes.

\begin{definition}[CSSSP] The constrained single-source shortest paths problem is defined by the following algorithm contract:
\begin{itemize}
    \item \textbf{Input:} An (unweighted, undirected) graph $G=(V,E)$, a set of ``gray" edges $E_g \subset E$, a source vertex $s \in V$, and a positive integer $g$.
    
    \item \textbf{Output:} For every $t \in V$, a path $P(s,t)$ on $\leq g$ gray edges, where $|P(s,t)| \leq |P'(s,t)|$ for all $s \leadsto t$ paths $P'(s,t)$ on $\leq g$ gray edges.
\end{itemize}
\end{definition}

Our modification to Chechik's construction will also make use of constrained shortest path finding, but the CSSSP problem is stronger than necessary for our purposes, and we can get away with a better runtime by solving a weaker problem. In this section, we define and give an efficient algorithm for a weaker variation on CSSSP, which we'll call weak CSSSP. In particular, we will only need to find constrained shortest paths from $s$ to $t$ in situations where a certain type of $s \leadsto t$ constrained path already exists. We define these paths and call them \textbf{g-short paths}.%, and the fastest algorithm for CSSSP we're aware of \cite{AliAbdi2019} is slower than plain SSSP.

\begin{definition}[g-short path]
For two nodes $s,t$, an $s \leadsto t$ path is called ``g-short" if it has $<g$ gray edges and is of the form $(s,s') \circ \pi_G(s',t') \circ (t,t')$, where $\pi_G(s',t')$ is a shortest path.
\end{definition}

\noindent
These g-short paths naturally arise in the analysis of our 4-additive spanner construction later in the paper. Note that g-short paths are not necessarily \textit{shortest} paths, but they are at most $+2$ edges longer than a shortest path (in the present unweighted case). Now we are ready to define weak CSSSP:

%{\color{red} remove $+2$ constant from problem}
\begin{definition}[Weak CSSSP] The weak constrained single-source shortest paths problem is defined by the following algorithm contract:

\begin{itemize}
    \item \textbf{Input:} An (unweighted, undirected) graph $G=(V,E)$, a set of ``gray" edges $E_g \subset E$, a source vertex $s \in V$, and a positive integer $g$.
    
    \item \textbf{Output:} For every $t \in V$, a path $P(s,t)$ on $\leq 5g$ gray edges, satisfying the following: %find a better name
    \begin{itemize}
        \item If there exists an $s\leadsto t$ \textbf{g-short path}, then the \textbf{g-optimality condition} holds: $|P(s,t)| \leq |P_g(s,t)|$ for any $s \leadsto t$ path $P_g(s,t)$ on $<g$ gray edges.
	\item If no such path exists, then $P(s, t)$ can be anything.
        
        %path $P'(s,t)$ on $<g$ gray edges and with $|P'(s,t)| \leq \dist_G(s,t) + 4$, $P(s,t)$ is an $s \leadsto t$ path with $|P(s,t)| \leq |P'(s,t)|$.
    \end{itemize}

\end{itemize}
\end{definition}

%{\color{red} informally describe output condition}

Informally, if there is a ``short enough" path from $s$ to $t$ with $<g$ gray edges, then the outputted path $P(s,t)$ has $\leq 5g$ gray edges and is shortest among all $s \leadsto t$ paths with $<g$ gray edges. Besides the constant factor on the gray-edge bound, the only difference between weak CSSSP and CSSSP is that we put a precondition on a g-short $s \leadsto t$ path existing. When we need weak CSSSP later in the paper, it will only be in situations where this g-short path exists.

One can solve the weak CSSSP problem by simply solving the general CSSSP problem. AliAbdi et al. in \cite{AliAbdi2019} present a label-setting algorithm ``Bi-SPP" (Bi-Colored Shortest Path Problem) for solving GB-SPP. This is the problem of finding shortest paths from a source node $s \in V$ to every other node $t \in V$ subject to the constraint that the paths have $\leq g$ gray \textbf{nodes}.

\begin{theorem}[Implicit in Prior Work]
There is an $\Oish(gm)$ solution to the weak CSSSP problem
\end{theorem}
\begin{proof}
%TODO; elucidate more
Given input to an instance of the weak CSSSP problem: we designate a \textit{node} to be gray if it is adjacent to a gray edge. It is clear that if a path has $< g$ gray edges, it must then have $< 2g$ gray nodes. Thus if we solve GB-SPP on this graph with parameter $g'=2g$, the resulting paths $P(s,t)$ satisfy $|P(s,t)| \leq |P_g(s,t)|$ for any path $P_g(s,t)$ on $<g$ gray edges. Furthermore, the paths $P(s,t)$ have $< 2g \leq 5g$ gray edges. Therefore these resulting paths satisfy the weak CSSSP requirements. Using the Bi-SPP algorithm, this can be done in $O(g(m+n \log n)) = \Oish(gm)$ time \cite{AliAbdi2019}.

\end{proof}

%\gbnoteinline{Conceptually the following part is really strong, shows that you are extending prior work in a useful way.  But it would be helpful to signpost what you're doing here before you start talking about AliAbdi.  Maybe even make this a theorem statement: "Theorem (Implicit in Prior Work): This is the problem of finding shortest paths from a vertex $s \in V$ to every other vertex $t \in V$ subject to the constraint that the paths have $\leq g$ gray nodes", use the proof to sketch the reduction to AliAbdi exactly as below.  Then afterwards you can follow up with your new improved algorithm.}
%In \cite{AliAbdi2019}, AliAbdi et al. present a label-setting algorithm ``Bi-SPP" (Bi-Colored Shortest Path Problem) for solving the ``Gray Vertices Bounded Single Source Shortest Paths Problem" (GB-SPP). This is the problem of finding shortest paths from a vertex $s \in V$ to every other vertex $t \in V$ subject to the constraint that the paths have $\leq g$ gray nodes. Their solution has an $O(g(m+n \log n)) = \Oish(gm)$ worst case runtime. Even though this constraint is on nodes and not edges, weak CSSSP can be solved with GB-SPP, as we will now describe.

We now present a new algorithm for solving weak CSSSP with the same runtime as plain SSSP in the weighted setting - using Dijkstra's algorithm with Fibonacci heaps, a runtime of $O(m + n\log n) = \Oish(m)$. We give each non-gray edge weight 1, and each gray edge weight $1+g^{-1}$. We run Dijkstra's algorithm with these weights and report the paths it computes. The intuition behind this approach is that we ``punish" gray edges by a value that ensures both the g-optimality condition and $\leq 5g$ gray edges in the paths Dijkstra reports. The punishment is big enough so that it is impossible for the reported path to have too many gray edges without violating the fact that a g-short $s\leadsto t$ path exists, and small enough so that the total punishment on the $P_g(s,t)$ paths described in the algorithm contract is less than $1$, effectively acting as a tiebreaker between such paths. %ask greg about this

\begin{figure}[h]
\begin{algorithm}[H]
    \label{alg:weak}\caption{Weak CSSSP via Edge Weighting}
    \KwIn{Undirected, unweighted, graph $G=(V,E)$\\
    Source vertex $s\in V$\\
    Set of gray edges $E_g \subseteq E$\\
    Positive integer $g$
    }\vspace{1em}
   
   %$\delta \leftarrow g^{-1}$
   
    \ForEach{edge $(u,v) \in E$}{
        $w(u,v)\leftarrow 1+g^{-1}$ if $(u,v)$ is gray, and $w(u,v)\leftarrow 1$ otherwise.
    } %could also split this into two loops; loop over all edges in E and add 1 to their weight, then separately loop
    % over edges of E_g to add delta.

    Run Dijkstra on $s$ with weight function $w$ to get paths $P(s,t)$ for each other $t \in V$\\
    \Return these $P(s,t)$ paths %is this good convention? Should the return statement be more explicit?
    % Greg: nah this is fine how it is

\end{algorithm}
\end{figure}

\begin{theorem}
Algorithm \ref{alg:weak} solves weak CSSSP in $O(m + n\log n) = \Oish(m)$ time.
\end{theorem}
\begin{proof}
The time complexity follows immediately from the complexity of Dijkstra's algorithm, which is the dominant stage of the algorithm. We now prove correctness.\\

%Let $t \in V$ and let $P'(s,t)$ be an $s \leadsto t$ path on $<g$ gray edges, and $|P'(s,t)| \leq \dist_G(s,t) + 2$. We will first show that $P(s,t)$ has $\leq 3g$ gray edges. Suppose to show a contradiction that it has $>3g$ gray edges. Note by construction that the weight of a path is its length plus $\delta$ times the number of gray edges. Thus  $w(P(s,t)) > |P(s,t)| + \delta \cdot 3(g+2) = |P(s,t)| + 3$. Furthermore, $w(P'(s,t)) \leq |P'(s,t)| + \delta \cdot (g+1) < |P'(s,t)| + 1$. But we also have, by the fact that $P(s,t)$ is the lowest-weight $s \leadsto t$ path, that $w(P(s,t)) \leq w(P'(s,t))$, and thus $|P(s,t)|+3 < |P'(s,t)| + 1 \leq \dist_G(s,t) + 3$. This implies that $|P(s,t)| < \dist_G(s,t)$, which is a contradiction.\\
\noindent
Let $t \in V$ and suppose a g-short $s\leadsto t$ path $P'(s,t)=(s,s') \circ \pi_G(s',t') \circ (t',t)$ exists. We will first show that $P(s,t)$ has $\leq 5g$ gray edges. Suppose to show a contradiction that it has $>5g$ gray edges. Note by construction that the weight of a path is its length plus $g^{-1}$ times the number of gray edges. Thus  $w(P(s,t)) > |P(s,t)| + g^{-1}\cdot 5g = |P(s,t)| + 5$. Furthermore, $w(P'(s,t)) < |P'(s,t)| + g^{-1} \cdot g = |P'(s,t)| + 1$. But we also have, by the fact that $P(s,t)$ is the lowest-weight $s \leadsto t$ path, that $w(P(s,t)) \leq w(P'(s,t))$, and thus 

\begin{align*}
    |P(s,t)|+5 &< |P'(s,t)| + 1\\
    &\leq |\pi_G(s',t')| + 2 + 1\\
    & \leq |\pi_G(s,t)| +2+ 2 + 1\\
    &= \dist_G(s,t)+5
\end{align*}
\noindent
This implies that $|P(s,t)| < \dist_G(s,t)$, which is a contradiction. Thus the computed path $P(s,t)$ has $\leq 5g$ gray edges. We now show the g-optimality condition to complete the proof: let $P_g(s,t)$ be an arbitrary $s\leadsto t$ path on $<g$ gray edges. We have that 

\begin{align*}
    |P(s,t)| &\leq w(P(s,t))\\
    &\leq w(P_g(s,t))\\
    & < |P_g(s,t)| + g\cdot g^{-1}\\
    &=|P_g(s,t)|+1
\end{align*}

\noindent
Thus $|P(s,t)| \leq |P_g(s,t)|$ as required.

\end{proof}

\subsection{Application to 4-Additive Spanner Construction}

We are now ready to state our modification to Chechik's spanner construction. Two insights allow us to improve the efficiency: (i) instead of finding the constrained shortest paths between the clusters, it is sufficient to only do this for paths between $S_2$ nodes. Furthermore, (ii) it is sufficient to compute the weak CSSSP paths for this task. The constraint we place on the paths we find will be a constraint on ``heavy edges", which are edges adjacent to heavy nodes

\begin{definition}[heavy edge]
An edge $(u,v)$ is called \textbf{heavy} if either $u$ or $v$ are heavy nodes (having degree $\geq \mu$).
\end{definition}

Besides these changes, the construction is the same as Chechik's original construction. We note that (i) alone was used by Ahmed et al. in their spanner construction for \textit{weighted} graphs \cite{ahmed2021additive}; combining (i) with AliAbdi et al.'s algorithm for CSSSP \cite{AliAbdi2019} yields an $\Oish(mn^{4/5})$ run time, though this was not known to the authors at the time [Personal Communication]. (ii) is a novel method and is responsible for the next polynomial step down to $\Oish(mn^{3/5})$. We now prove our main result through the following series of lemmas: %except for the step of finding constrained shortest paths between clusters. We instead compute weak CSSSP for each node of $S_2$ and add the resulting paths between $S_2$ nodes.

\begin{figure}[h]
\begin{algorithm}[H]
    \label{alg:best}\caption{Faster 4-Additive Spanner Construction}
    \KwIn{$n$-node graph $G=(V,E)$}\vspace{1em}
    $E'=$ All edges incident to light nodes\\
    Sample a set of nodes $S_1$ at random, every node with probability $9\mu/n$\\
    \ForEach{node $x\in S_1$}{
        Construct a BFS tree $T(x)$ rooted at $x$ spanning all vertices in $V$\\
        $E'=E'\cup E(T(x))$\\
    }\vspace{1em}
    Sample a set of nodes $S_2$ at random, every node with probability $1/\mu$\\
    \ForEach{heavy node $x$ so that $(\{x\}\cup\Gamma_G(x))\cap S_2 = \varnothing$}{
        Add all incident edges of $x$ to $E'$\\
    }
    % \ForEach{node $x\in S_2$}{
    %     C(x) = \{x\}
    % }
    \ForEach{heavy node $v$ so that $v\not\in S_2$ and $\Gamma_G(v)\cap S_2\neq\varnothing$}{
        Arbitrarily choose one node $x\in\Gamma_G(v)\cap S_2$\\
        % $C(x) = C(x) \cup \{x\}$\\
        $E' = E' \cup \{(x,v)\}$\\
    }
    % \ForEach{edge $(u,v) \in E$}{
    %     $w(u,v)\leftarrow 1+\delta$ if both $u,v$ are heavy nodes, and $w(u,v)\leftarrow 1$ otherwise.
    % }
    
    \ForEach{node $x_1 \in S_2$}{
        Compute weak CSSSP on $G$, with $g = \mu^3/n+2$, $x$ as the source vertex, and $E_g$ as the set of heavy edges, to get paths $P(x_1,x_2)$ for each $x_2 \in V$.\\
        Add $E(P(x_1,x_2))$ to $H$ for each $x_2 \in S_2$
    }

    \Return $H=(V,E')$
\end{algorithm}
\end{figure}

%After constructing the $S_2$ clusters, 

% say algorithm, prove correctness etc
\noindent
The first lemma is standard:

\begin{lemma}[\cite{chechik2013new, Baswana2010}]
For a shortest path $\pi_G(s,t)$ in a graph $G$ and for any vertex $v \in V(G)$, $v$ has $\leq 3$ neighbors in $\pi_G(s,t)$.
\end{lemma}
\begin{proof}
Suppose to show a contradiction that $v$ has four neighbors $u_1,u_2,u_3,u_4 \in \pi_G(s,t)$, and assume WLOG $u_1 \prec u_2 \prec u_3 \prec u_4$ in $\pi_G(s,t)$. This implies $\dist_G(u_1,u_4) \geq 3$. But the path $(u_1,v) \circ (v,u_4)$ has length $2$, which is a contradiction.
\end{proof}

We now show that for any two nodes $s,t$ of $H$, we have with very high probability that $\dist_H(s,t) \leq \dist_G(s,t)+4$. Note that it's sufficient to prove this result when $s,t$ don't have all of their edges included in $H$. We call such nodes ``uncovered". This is because when $s,t$ are not both uncovered, it's enough to demonstrate this stretch for the subpath of $\pi(s,t)$ beginning and ending at the first and last uncovered nodes respectively.

\begin{definition}
A node $v \in V(G)$ is said to be ``covered" in $H$ if all of its edges are included in $H$.
\end{definition}

This is because if $s,t$ were not both covered, we could let $s'$ be the first uncovered node on $P(s,t)$ and $t'$ the last uncovered node, and then $\dist_H(s',t') \leq \dist_G(s',t')+4 \implies \dist_H(s,t) \leq \dist_G(s,t)+4$. This is because all the edges of $P(s,t)$ occurring before $s'$ or after $t'$ are already in $H$, so any stretch must occur between $s'$ and $t'$. This allows us to assume that $s,t$ are in $\Gamma_G(S_2)$, as all other nodes are covered by our algorithm. The proof of the following lemma is identical to the first part of the proof of Lemma 2.2 in \cite{chechik2013new}, but we repeat it here for completeness.

%Is this exposition unnecessary?
% Greg: I kind of like what you wrote in this commented part, I'd include it

\begin{lemma}
For any two uncovered nodes $s,t\in V(G)$ such that the canonical shortest path $\pi_G(s,t)$ has $>\mu^3/n$ heavy nodes, we have $\dist_H(s,t) \leq \dist_G(s,t)+4$ with probability $\geq 1-\frac{1}{n^3}$
\end{lemma}
\begin{proof}
In this case, we claim that there is a $\geq 1-1/n^3$ probability that $\pi_G(s,t)$ is adjacent to a BFS tree in $H$. $\pi_G(s,t)$ has $>\mu^3/n$ heavy nodes, each of degree $\geq \mu$. Thus the sum of the degrees of nodes on $\pi_G(s,t)$ is $> \mu^4/n$. By Lemma 5, this implies there are at least $\mu^4/3n$ nodes adjacent to $\pi_G(s,t)$. Each node $v$ has probability $9\mu/n$ of being included in $S_1$, and thus having a shortest-path tree rooted at $v$ in $H$. Therefore the probability that \textit{none} of these nodes adjacent to $\pi_G(s,t)$ have such a tree rooted at them is
   \begin{align*}
       &\leq (1-9\mu/n)^{\mu^4/3n}\\
       &=(1-9\log^{1/5}n/n^{3/5})^{n^{3/5}\log^{4/5}n/3}\\
       &=(1-9\log^{1/5}n/n^{3/5})^{(n^{3/5}/9\log^{1/5}n)\cdot 3\log n}\\
       &\leq \left(\frac{1}{e}\right)^{3 \log n}\\
       &\leq 1/n^3
   \end{align*}
   where we used the fact that $(1-\frac{1}{x})^x < 1/e$ for $x\geq 1$. Thus, we have a $> 1-1/n^3$ probability of the existence of a node $r$ neighboring some $u \in \pi_G(s,t)$ such that a BFS tree rooted at $r$ is in $H$. When this is the case, we can simply take the $s\leadsto r$ followed by the $r \leadsto t$ shortest paths provided by the BFS tree, which has a stretch factor of 2 as shown below:
 \begin{align*}
    \dist_H(s,t) &\leq \dist_H(s,r) + \dist_H(r,t)\\
    &= \dist_G(s,r) + \dist_G (r,t)\\
    &\leq \dist_G(s,u) + 1 + \dist_G (u,t)+1\\
    &= \dist_G(s,t)+2
   \end{align*}

\end{proof}
%\gbnoteinline{minor tex idiom: if you end a proof with math in the align enviornment, you can add \textbackslash{}tag$^*$\{\textbackslash{}qedhere\} right after the last equation to force the square to align with the math}

\begin{lemma}
For any two uncovered nodes $s,t\in V(G)$ such that the canonical shortest path $\pi_G(s,t)$ has $\leq \mu^3/n$ heavy nodes, we have $\dist_H(s,t) \leq \dist_G(s,t)+4$ with probability $\geq 1-\frac{1}{n^3}$.
\end{lemma}
\begin{proof}
Both $s,t$ are uncovered and thus in $\Gamma_G(S_2)$. Let $x_1,x_2 \in S_2$ such that $s \in \Gamma_H(x_1)$ and $t \in \Gamma_H(x_2)$. We assume $x_1 \neq x_2$ as this case is trivial.\\

%Furthermore... might want more explanation?
%Greg: for which part? this proof is decently clear I think
Call an edge $(u,v)$ ``heavy" if both $u$ and $v$ are heavy nodes. Since $\pi_G(s,t)$ is a (shortest) path on $\mu^3/n$ heavy nodes, it has $<\mu^3/n = g-2$ heavy edges, meaning $P'(x_1,x_2) := (x_1,s) \circ \pi_G(s,t) \circ \pi_G(t,x_2)$ has $<g$ gray edges and is thus a g-short path\footnote{If $x_1=s$ or $x_2=t$, $(x_1,s)$ (resp. $(t,x_2)$) are by convention empty paths.}. Thus when we compute the weak CSSSP path $P(x_1,x_2)$, we have that $|P(x_1,x_2)| \leq |P'(x_1,x_2)|$ . Furthermore, $|P'(x_1,x_2)| \leq \pi_G(s,t)+2 \leq \dist_G(x_1,x_2)+4$, since $\pi_G(s,t)$ is a shortest path in $G$. Let $P(s,t) = (s,x_1) \circ P(x_1,x_2) \circ (x_2,t)$, and note that $P(s,t)$ is a path in $H$. This path witnesses that

\begin{align*}
    \dist_H(s,t) &\leq 2 + |P(x_1,x_2)|\\
    &\leq 2 + |P'(x_1,x_2)|\\
    &\leq 2 + 2 + |\pi_G(s,t)|\\
    &= \dist_G(s,t)+4
\end{align*}

Thus we deterministically have the required stretch for such node pairs.

\end{proof}

Correctness follows by the above two lemmas and the union bound, which gives us that with probability $>1-1/n$, $\dist_H(s,t) \leq \dist_G(s,t) + 4$ holds for all $s,t \in V(G)$.

\begin{lemma}
The subgraph $H$ produced by Algorithm \ref{alg:best} has $O(n\mu) = \Oish(n^{7/5})$ edges with high probability. %wording?
%Greg: this is the right wording of WHP if that's the concern.
\end{lemma}
\begin{proof}
We can separate the addition of edges to $H$ into 4 types:

\begin{enumerate}
    \item The edges incident to light nodes are added. Each light node is incident to $\leq \mu$ edges by definition, so $O(n\mu) = \widetilde{O}(n^{7/5})$ edges are added.
    
    \item The BFS tree of each node in $S_1$ is added. Each such tree contributes $O(n)$ edges. The probability of a node being added to $S_1$ is $9\mu/n$, so $|S_1|=\Theta(\mu)$ with high probability, and thus $O(\mu \cdot n) = \widetilde{O}(n^{7/5})$ edges are added with high probability.%\gbnote{Not quite right -- $|S_1|=9\mu$ (exactly) is a pretty low-probability event, but $|S_1| = \Theta(\mu)$ holds with high probability.  Gotta use the Thetas}
    
    \item Edges adjacent to heavy nodes $v$ that are  $\notin \Gamma_G(S_2)$ are added. Nodes are added to $S_2$ with probability $1/\mu$, thus the probability of $v$ being neither in $S_2$ nor adjacent to a node in $S_2$ is $\leq (1-1/\mu)^{\deg(v)+1}$. If $\deg(v) = \Omega(\mu \log n)$, then it is adjacent to a node in $S_2$ with high probability. Thus the number of edges added for $v$ is at most $1+\deg(v)(1-1/\mu)^{\deg(v)} < \mu$ with high probability. Unioning over all $v$, this adds $O(n\mu)=\widetilde{O}(n^{7/5})$ edges with high probability.
    
    \item Edges on paths between $S_2$ nodes with $\leq 5\mu^3/n$ heavy edges are added. $|S_2| = \Theta(n/\mu)$ with high probability, yielding $\Theta(n^2/\mu^2)$ pairs of $S_2$. All the light edges (edges adjacent to light nodes) have already been added to $H$, so each path between these $S_2$ pairs adds at most $\Theta(\mu^3/n)$ edges. Unioning over the number of pairs, this adds $O(\mu n)=\widetilde{O}(n^{7/5})$ edges with high probability. \qedhere
\end{enumerate}

%Change expected -> high probability

% Now we show that this edge bound is reached with very high probability. For $v \in V$, let $X_v,Y_v$ be indicator random variables, with $X_v=1$ if $v \in S_1$ and $Y_v = 1$ if $v \in S_2$. Let $X = \sum_{v \in V}X_v$ and $Y = \sum_{v \in V}Y_v$. We have from the above that $E[X]=9\mu$ and $E[Y]=n/\mu$. By the Chernoff bound we have, for any $\epsilon>0$,

% $$\Pr[|X-9\mu| \geq 9\mu \cdot \epsilon] \leq 2e^{-(9\mu)\epsilon^2/3}$$
% $$\Pr[|Y-n/\mu| \geq n/\mu \cdot \epsilon] \leq 2e^{-(n/\mu)\epsilon^2/3}$$

% If we pick $\epsilon_1 = k/\sqrt{9\mu}$ and $\epsilon_2 = k/\sqrt{n/\mu}$ (where $k$ is an unspecified constant), then

% $$\Pr[|X-9\mu| \geq k\sqrt{9\mu}] \leq 2e^{-k^2/3}$$
% $$\Pr[|Y-n/\mu| \geq k\sqrt{n/\mu}] \leq 2e^{-k^2/3}$$

% By picking $k$ to be an arbitrarily large constant, these probabilities can be made to be arbitrarily small. This shows that asymptotically, the algorithm produces a graph on $\Oish(n^{7/5})$ edges with arbitrarily high probability.

\end{proof}

\begin{lemma}
On an $n-$node $m-$edge input graph, Algorithm \ref{alg:best} runs in $\Oish(mn^{3/5})$ with high probability.
\end{lemma}
\begin{proof}

The only two superlinear stages of the algorithm are (a) the generation of the $S_1$ Breadth-First Search trees, and (b) solving weak CSSSP for each node of $S_2$. For (a): nodes are sampled to be in $S_1$ with probability $9\mu/n$, so $|S_1| = O(\mu) = \Oish(n^{2/5})$ with high probability. BFS has worst-case runtime $O(m)$. Thus this stage is $O(m\mu)=\Oish(mn^{2/5})$ time. For (b): we showed in section 2.1 an algorithm that solves weak CSSSP in $\Oish(m)$ time, which we run for each node of $S_2$. Multiplying this over the size of $S_2$ (which has size $\Oish(n^{3/5})$ with high probability), we get $\Oish(mn^{3/5})$ time with high probability.

%Just reference chernoff bound analysis in previous lemma?
%Greg: You don't need to restate any Chernoff bound stuff in this lemma, or re-explain that Chernoff was used to get the WHP guarantee in the last lemma

\end{proof}

\noindent
Theorem 1 now follows from Lemmas 6-9.

\section{Weighted +4 Spanner}

%\gbnoteinline{General Comment: You don't necessarily have to re-prove proofs if they're mostly the same as the weighted section.  It's potentially fine to say "this is basically the same as [previous lemma] with [minor changes]," if you feel like that's enough for the reader to recover the proof.  It's up to you though, nothing wrong with full detail}

In this section, we prove Theorem 2 by first synthesizing a weighted analogue of the weak CSSSP problem from section 2.1, then we apply this in a similar fashion to create our $+4W(s,t)+\epsilon W$ construction. We note that a $+4W$, $\Oish(n^{7/5})$ edge, $\Oish(mn^{3/5})$ time spanner construction is a very straightforward generalization of the methods presented in the unweighted part of this paper. However, $W(s,t)$ error is preferable (except in the case when all the edge weights are equal).

%Summarize methods of ahmed2021additive? 
%Greg: nah, probably distracting at this point.

\subsection{Weighted Constrained Shortest Paths}

In our first step towards our weighted spanner, we will generalize weak CSSSP from section 2.1 to the weighted setting, and give a Dijkstra-time algorithm for solving it. This is also where the $\epsilon$ error of the construction will be incurred, so our generalization will incorporate an error parameter. There are many ways to generalize weak CSSSP to the weighted setting, but we choose this one as it is what we've been able to find a use for in the weighted spanner construction. %explain why we need error, explain why global and not local.

\begin{definition}[Weighted Weak CSSSP (with error)] The weighted weak constrained single-source shortest paths problem with error is defined by the following algorithm contract:

\begin{itemize}
    \item \textbf{Input:} An (undirected) graph $G=(V,E)$, a set of ``gray" edges $E_g \subset E$, a source vertex $s \in V$, a weight function $w:V\to \mathbb{R}^+$, an error parameter $1>\epsilon>0$, and a positive integer $g$.
    
    \item \textbf{Output:} For every $t \in V$, a path $P(s,t)$ on $\leq 5g/\epsilon$ gray edges, satisfying the following: %find a better name
    \begin{itemize}
        \item If an $s\leadsto t$ g-short \footnote{We reuse the same definition from section 2.1.} path exists, then the \textbf{near g-optimality condition} holds: $w(P(s,t)) < w(P_g(s,t)) + \epsilon W$ for any $s\leadsto t$ path $P_g(s,t)$ on $<g$ gray edges, where $W = \max_{e \in E}w(E)$.
        
        \item If no such paths exist, $P'(s,t)$ can be anything.
        
        %path $P'(s,t)$ on $<g$ gray edges and with $|P'(s,t)| \leq \dist_G(s,t) + 4$, $P(s,t)$ is an $s \leadsto t$ path with $|P(s,t)| \leq |P'(s,t)|$.
    \end{itemize}

\end{itemize}
\end{definition}

%clunky wording?
%Greg: seems ok to me
The only difference between this version of the problem and the unweighted version is that it is in terms of path weight instead of path length, and it allows an $\epsilon$ error in terms of the heaviest edge of the graph. As before, one can use AliAbdi et al's Bi-SPP \cite{AliAbdi2019} to solve this, but it is stronger than necessary. A modification to Algorithm \ref{alg:weak} gives us the fastest solution. It is the same as before, except that we now punish each gray edge by $+\epsilon W g^{-1}$ instead of $+g^{-1}$.

\begin{figure}[h]
\begin{algorithm}[H]
    \label{alg:deltaweighted}\caption{Weighted Weak CSSSP with Error by Edge Punishing}
    \KwIn{Undirected graph $G=(V,E)$\\
    Source vertex $s\in V$\\
    Set of gray edges $E_g \subseteq E$\\
    Weight function $w:E \to \mathbb{R}^{+}$\\
    Error $1>\epsilon>0$\\
    Positive integer $g$
    }\vspace{1em}
   
   Let $W = \max_{e \in E} w(e)$

    \ForEach{edge $(u,v) \in E$}{
        $w'(u,v)\leftarrow w(u,v)+\epsilon W g^{-1}$ if $(u,v)$ is gray, and $w'(u,v)\leftarrow w(u,v)$ otherwise.
    }

    Run Dijkstra's algorithm on $G$ with source node $s$ and weight function $w'$ to get paths $P(s,t)$ for each other $t \in V$\\
    \Return these $P(s,t)$ paths 

\end{algorithm}
\end{figure}

%put some syntactic filler?
%Greg: not necessary, but you could

\begin{theorem}
Algorithm \ref{alg:deltaweighted} solves Weighted Weak CSSSP  with error in $O(m + n \log n) = \Oish(m)$ time.
\end{theorem}
\begin{proof}

The time complexity follows immediately from the complexity of Dijkstra's algorithm, which is the dominant stage of the algorithm. We now prove correctness.

Let $t \in V$ and suppose a g-short $s\leadsto t$ path $P'(s,t)=(s,s') \circ \pi_G(s',t') \circ (t',t)$ exists. We will first show that $P(s,t)$ has $\leq 5g/ \epsilon$ gray edges. Suppose to show a contradiction that it has $>5g/\epsilon$ gray edges. Thus

\begin{align*}
    w(P(s,t)) + 5W &< w'(P(s,t))\\
    &\leq w'(P'(s,t))\\
    &< w(P'(s,t)) + g\cdot (\epsilon W g^{-1}) = w(P'(s,t)) + \epsilon W\\
    &\leq w(P'(s,t)) + W\\
    &\leq w(\pi_G(s',t'))+2W+W
\end{align*}

Therefore $w(P(s,t)) + 4W < w(\pi_G(s',t')) + 2W$. Since $\pi_G(s',t')$ is a shortest path, $w(\pi_G(s',t')) \leq w(P(s,t)) + 2W$. This implies $w(P(s,t)) +4W < w(P(s,t)) + 4W$, a contradiction.\\

Now we prove near g-optimality. Let $P_g(s,t)$ be an arbitrary $s\leadsto t$ path on $<g$ gray edges. Then 

\begin{align*}
w(P(s,t)) &\leq w'(P(s,t))\\
&\leq w'(P_g(s,t))\\ 
&< w(P_g(s,t)) + g \cdot (\epsilon W g^{-1})\\
&= w(P_g(s,t)) + \epsilon W
\end{align*}

as required.

\end{proof}

\subsection{$+4W(\cdot,\cdot)+\epsilon W$ Spanner}

In this section we modify our Section 2 construction to construct $+4W(\cdot,\cdot) + \epsilon W$ spanners, for any $\epsilon > 0$, on $\Oish_\epsilon(n^{7/5}e^{-1})$ edges in $\Oish(mn^{3/5})$ time. This will make use of our weak CSSSP generalization to the weighted setting with error discussed in the previous subsection, which is where the $+\epsilon W$ stretch is incurred.%is it good to have this last sentence here, or should it just be in the intro for this chapter?
%Greg: this is worth mentioning in the intro of the entire paper, once you add the weighted part, rather than here or in the beginning of the section.  

%should we define d-lightweight initialization first, or just this is ok?
%Greg: this would be ok, *but* you might have to define d-lightweight initialization formally at some point in order to state the neighborhood lemma from Ahmed, which I think you need in your proof
The construction differs from Algorithm \ref{alg:best} in three important ways: (i): instead of adding all the edges of light nodes to the spanner $H$, we perform a \textbf{$\mu$-lightweight initialization} - that is, we add the $\mu$ lightest edges of each node to $H$ (breaking ties arbitrarily), a technique introduced in \cite{ahmed2021additive}.

\begin{definition}[$d-$lightweight initialization \cite{ahmed2021additive}]
A $d-$lightweight initialization $H = (V,E')$ of a weighted graph $G = (V,E)$ is a subgraph created by selecting the $d$ lightest edges of every node of $G$, breaking ties arbitrarily.
\end{definition}

(ii): Instead of computing standard weak CSSSP (as seen in Section 2.2) on the $S_2$ nodes, we compute our new weighted version with error. (iii): We now omit the step of connecting ``heavy nodes" with nodes of $S_2$. Instead we rely on these connections happening ``naturally" due to our $\mu$-lightweight initialization, which allows us to make use of the properties the lightweight initialization gives us.

%put epsilon as input?
%Greg: yeah, do this
\begin{figure}[h]
\begin{algorithm}[H]
    \label{alg:weightedspan}\caption{$+4W(\cdot,\cdot) + \epsilon W$ Spanner}
    \KwIn{$n$-node graph $G=(V,E)$\\
	Weight function $w:E \mapsto \mathbb{R}^+$\\
	Error parameter $1>\epsilon>0$\\
	}\vspace{1em}
    Let $H_0=(V,E')$ be a $\mu$-lightweight initialization of $G$.\\
    
    Sample a set of nodes $S_2$ at random, every node with probability $1/\mu$\\
    \ForEach{node $x$ so that $(\{x\}\cup\Gamma_{H_0}(x))\cap S_2 = \varnothing$}{
        Add all incident edges of $x$ to $E'$\\
    }
    Sample a set of nodes $S_1$ at random, every node with probability $9\mu/n$\\
    \ForEach{node $x\in S_1$}{
        Construct a shortest-path tree $T(x)$ rooted at $x$ spanning all vertices in $V$\\
        $E'=E'\cup E(T(x))$\\
    }\vspace{1em}
    %
    % \ForEach{node $x\in S_2$}{
    %     C(x) = \{x\}
    % }
    % \ForEach{heavy node $v$ so that $v\not\in S_2$ and $\Gamma_G(v)\cap S_2\neq\varnothing$}{
    %     Arbitrarily choose one node $x\in\Gamma_G(v)\cap S_2$\\
    %     % $C(x) = C(x) \cup \{x\}$\\
    %     $E' = E' \cup \{(x,v)\}$\\
    % }
    % \ForEach{edge $(u,v) \in E$}{
    %     $w(u,v)\leftarrow 1+\delta$ if both $u,v$ are heavy nodes, and $w(u,v)\leftarrow 1$ otherwise.
    % }
    
    \ForEach{node $x_1 \in S_2$}{
        Compute weighted weak CSSSP with error on $G$, with $g = \mu^3/n+2$, $x_1$ as the source vertex, $E_g = E \setminus E'$, and $\epsilon$ as the error parameter, to get paths $P(x_1,x_2)$ for each $x_2 \in V$.\\
        Add $E(P(x_1,x_2))$ to $E'$ for each $t \in S_2$
    }

    \Return $H=(V,E')$
\end{algorithm}
\end{figure}

\noindent
Note that $H_0$ in Algorithm 5 denotes the $\mu-$initialization of $G$, before any additional edges are added to $E'$. We will refer to $H_0$ in our proofs as it will allow us to make use of the fact that the edges added are only the ones form the lightweight initialization. We now prove correctness by the following series of lemmas. We will make use of the following theorem due to Ahmed et al.

%%%Might not actually need this... (scratch that, we do)

\begin{theorem}[\cite{ahmed2020weighted}]
Let $H$ be a $d-$lightweight initialization of an undirected, weighted graph $G$. Then if a shortest path $\pi_G(s,t)$ is missing $l$ edges in $H$, there are $\Omega(dl)$ different nodes adjacent to $\pi_G(s,t)$ in $H$.
\end{theorem}

%Need H_0 to distinguish between edges that come from initialization and edges added later.

\begin{lemma}
For any two nodes $s,t\in V(G)$ such that the canonical shortest path $\pi_G(s,t)$ is missing $>\mu^3/n$ edges in $H_0$, we have $\dist_H(s,t) \leq \dist_G(s,t)+4W(s,t)$ with probability $\geq 1-\frac{1}{n^3}$
\end{lemma}

\begin{proof}

This implies that there are $>\mu^3/n$ \textit{nodes} along $\pi_G(s,t)$ with missing incident edges in $H_0$. Call the set of these nodes $S$. Since $H_0$ is a $\mu-$lightweight initialization, it follows that nodes missing an edge in $H_0$ must have degree $>\mu$ in $G$. We utilize the above theorem from \cite{ahmed2020weighted}, which allows us to conclude that there are $\Omega(\mu^4/n)$ different nodes adjacent to $\pi_G(s,t)$ in $H_0$. 

%By Lemma 3, $S$ then has $>|S|\mu/3 > \mu^4/n$ unique neighbors in $H_0$.
%\gbnoteinline{The proof of Lemma 3 is unweighted only.  You need to cite the corresponding Lemma from Ahmed to make this work, I think (basically the commented line below)}

As shown in Lemma 4 (the details of which we won't repeat), we have a $> 1-1/n^3$ probability of one of these neighbors being a member of $S_1$. In this case, let $r\in S_1$ be adjacent to a node $q$ of $\pi_G(s,t)$ in $H_0$. Since $(q,r) \in H_0$ but $q$ is disconnected from $\pi_G(s,t)$ in $H_0$, it follows from the fact that $H_0$ is a $\mu-$lightweight initialization that $w(q,r)$ is lighter than the missing edge incident to $q$ on this path, i.e. $w(q,r) \leq W(s,t)$.

Since $r \in S_1$, it is the root of a shortest path tree of $G$ included in $H$. Thus, since $\pi_G(s,r)$ is a shortest path, $w(\pi_G(s,r)) \leq  w(\pi_G(s,q)) + w(q,r) \leq w(\pi_G(s,q))+W(s,t)$. Likewise, $w(\pi_G(r,t)) \leq w(\pi_G(q,t)) + W(s,t)$. Thus, the path $\pi_G(s,r) \circ \pi_G(r,t)$, which belongs to $H$, witnesses that

\begin{align*}
\dist_H(s,t) &\leq w(\pi_G(s,q))+W(s,t) + w(\pi_G(q,t)) + W(s,t)\\
&= w(\pi_G(s,t)) + 2W(s,t)\\
&= \dist_G(s,t) + 2W(s,t)
\end{align*} 

Thus with probability $>1-1/n^3$, we achieve the required stretch.

\end{proof}

\begin{lemma}
For any two nodes $s,t\in V(G)$ such that the canonical shortest path $\pi_G(s,t)$ is missing $< \mu^3/n$ edges in $H_0$, we have $\dist_H(s,t) \leq \dist_G(s,t)+4W(s,t) + \epsilon W$.
\end{lemma}
\begin{proof}
Just as in the unweighted case, we can assume WLOG that $s$ and $t$ are uncovered, and thus are each in the neighborhood of an $S_2$ node (all nodes $\notin \gamma(S_2)$ are covered by the algorithm). Let $x_1,x_2 \in S_2$ such that $s \in \Gamma(x_1)$ and $t \in \Gamma(x_2)$. We can furthermore assume that the first and last edges of $\pi_G(s,t)$ are missing from $H_0$, otherwise we could simply push our analysis to the first/last nodes on $\pi_G(s,t)$ to be severed from the path in $H_0$.

Since $\pi_G(s,t)$ is a shortest path with $< g-2 = \mu^3/n$ missing (gray) edges, $P'(x_1,x_2):=(x_1,s')\circ \pi_G(s',t') \circ (t',x_2)$ is a g-short path. Thus by weighted weak CSSSP with error, our construction yields a path $P(x_1,x_2)$ with $w(P(x_1,x_2)) \leq w(P'(x_1,x_2)) + \epsilon W$. Furthermore, by the fact that $H_0$ is a $\mu$-lightweight initialization, the edge $(x_1,s)$ must be lighter than the first edge of $\pi_G(s,t)$, and $(t,x_2)$ must be lighter than the last edge of $\pi_G(s,t)$. Thus $w(x_1,s),w(t,x_2) \leq W(s,t)$.\\

Therefore, $w(P(x_1,x_2)) \leq w(\pi_G(s,t))+ 2W(s,t) + \epsilon W$. Thus, the path $(s,x_1) \circ P(x_1,x_2) \circ (x_2,t)$ in $H$ witnesses that

\begin{align*}
\dist_H(s,t) &\leq w(s,x_1) + w(P(x_1,x_2)) + w(x_2,t)\\
&\leq 2W(s,t) + w(P(x_1,x_2))\\
&\leq 2W(s,t) + 2W(s,t) + \epsilon W + w(\pi_G(s,t))\\
&= \dist_G(s,t) + 4W(s,t) + \epsilon W
\end{align*}
\noindent
Thus we deterministically have the required stretch for such node pairs.

\end{proof}

\noindent
Correctness now follows by the above two lemmas and the union bound.

\noindent
We now show that $H$ has the desired edge bound. We note that the details of this proof are mostly the same as the corresponding proof for our unweighted construction.

\begin{lemma}
$H$ has (with high probability) $\Oish_\epsilon(n^{7/5}\epsilon^{-1})$ edges.
\end{lemma}
\begin{proof}
Algorithm \ref{alg:weightedspan} adds edges to $H$ in 4 stages:
\begin{enumerate}[(i)]
    \item We begin with a $\mu-$lightweight initialization of $G$, which adds $\Oish(n^{2/5})$ edges per node, giving a total of $\Oish(n^{7/5})$ edges. This covers all light nodes (of degree $<\mu$).

    %Note: same as lemma 6 (ii)
    \item We add the edges of a shortest path tree for each node of $S_1$. The probability of a node's inclusion into $S_1$ is $9\mu/n$, thus $|S_1| = \Oish(n^{2/5})$ with high probability. Adding $O(n)$ edges for each of the $S_1$ nodes thus yields $\Oish(n^{7/5})$ edges added with high probability.

    %Note: same as lemma 6 (iii)
    \item We add all the edges of heavy nodes $v$ (of degree $>\mu$) not in the neighborhood of any $S_2$ node. Nodes are added to $S_2$ with probability $1/\mu$, thus the probability of $v$ being neither in $S_2$ nor adjacent to a node in $S_2$ is $\leq (1-1/\mu)^{\deg(v)+1}$. If $\deg(v) = \Omega(\mu \log n)$, then it is adjacent to a node in $S_2$ with high probability. Thus the number of edges added for $v$ is at most $1+\deg(v)(1-1/\mu)^{\deg(v)} < \mu$ with high probability. Unioning over all $v$, this adds $O(n\mu)=\widetilde{O}(n^{7/5})$ edges with high probability.
    
    \item Edges on paths between $S_2$ nodes with $\leq 5\epsilon^{-1}\mu^3/n$ additional edges are added to $H$. $|S_2| = \Theta(n/\mu)$ with high probability, yielding $\Theta(n^2/\mu^2)$ pairs of $S_2$. Since each path between these pairs incurs $\leq 5\epsilon^{-1}\mu^3/n$ extra edges, this adds $O_\epsilon(\mu n \epsilon^{-1})=\widetilde{O}_\epsilon(n^{7/5}\epsilon^{-1})$ edges with high probability.
\end{enumerate}
\end{proof}

Finally, we show that Algorithm \ref{alg:weightedspan} has the desired runtime.

%This is also similar to the corresponding lemma for the unweighted chapter
\begin{lemma}
On an $n-$node $m-$edge input graph, Algorithm \ref{alg:weigtedspan} runs in $\Oish(mn^{3/5})$ with high probability.
\end{lemma}
\begin{proof}

The only two superlinear stages of the algorithm are (a) the generation of the $S_1$ shortest-path trees, and (b) solving weak weighted CSSSP with error for each node of $S_2$. For (a): nodes are sampled to be in $S_1$ with probability $9\mu/n$, so $|S_1| = O(\mu) = \Oish(n^{2/5})$ with high probability. Dijkstra has worst-case runtime $\Oish(m)$. Thus this stage is $\Oish(mn^{2/5})$ time. For (b): we showed in section 3.1 an algorithm that solves weak CSSSP in $\Oish(m)$ time, which we run for each node of $S_2$. Multiplying this over the size of $S_2$ (which has size $\Oish(n^{3/5})$ with high probability), we get $\Oish(mn^{3/5})$ time.

%Just reference chernoff bound analysis in previous lemma?

\end{proof}

\noindent
By Lemmas 12-15, we have now proven the main result of this chapter (Theorem 2).

\section{Conclusion}

In this paper we have presented a new state-of-the-art $\Oish(mn^{3/5})$ complexity result for constructing the +4 spanner, doing so by solving a novel pathfinding problem (weak CSSSP). This fills in a literature gap that has existed between +2,+6, and +8 spanners, as this is the first paper studying the efficiency of the +4 spanner construction. We also extended our methods to the weighted setting, where we were able to to derive a construction for $+4W(s,t)+\epsilon W$ spanners, with the same runtime and an $\epsilon^{-1}$ stretch to spanner size.

We believe that to find further polynomial time improvements to our construction would require a polynomial reduction in the number of $S_2$ nodes to compute shortest path trees on. The next bottleneck to the algorithm is the time needed to build the BFS trees rooted at the $S_1$ nodes, which is $\Oish(mn^{2/5})$ with high probability. For the weighted spanner construction, we hope to see a reduction in error to $(4+\epsilon)W(s,t)$ or $+4W(s,t)$ without a compromise to runtime, eliminating global error entirely.

\section*{Acknowledgements}

Many thanks to Greg Bodwin, without whose guidance this paper would not be possible. I also thank fellow students Eric Chen and Cheng Jiang, with whom I discussed an earlier version of the paper.

\bibliographystyle{alpha}
\bibliography{main}

\begin{thebibliography}{10}

\bibitem{ahmed2021additive}
{\sc Ahmed, R., Bodwin, G., Hamm, K., Kobourov, S., and Spence, R.}
\newblock On additive spanners in weighted graphs with local error.
\newblock {\em arXiv preprint arXiv:2103.09731\/} (2021).

\bibitem{ahmed2020graph}
{\sc Ahmed, R., Bodwin, G., Sahneh, F.~D., Hamm, K., Jebelli, M. J.~L.,
  Kobourov, S., and Spence, R.}
\newblock Graph spanners: A tutorial review.
\newblock {\em Computer Science Review 37\/} (2020), 100253.

\bibitem{ahmed2020weighted}
{\sc Ahmed, R., Bodwin, G., Sahneh, F.~D., Kobourov, S., and Spence, R.}
\newblock Weighted additive spanners.
\newblock In {\em International Workshop on Graph-Theoretic Concepts in
  Computer Science\/} (2020), Springer, pp.~401--413.

\bibitem{aingworth1999fast}
{\sc Aingworth, D., Chekuri, C., Indyk, P., and Motwani, R.}
\newblock Fast estimation of diameter and shortest paths (without matrix
  multiplication).
\newblock {\em SIAM Journal on Computing 28}, 4 (1999), 1167--1181.

\bibitem{AliAbdi2019}
{\sc AliAbdi, A., Mohades, A., and Davoodi, M.}
\newblock Constrained shortest path problems in bi-colored graphs: a
  label-setting approach.
\newblock {\em GeoInformatica\/} (Dec 2019).

\bibitem{alman2021refined}
{\sc Alman, J., and Williams, V.~V.}
\newblock A refined laser method and faster matrix multiplication.
\newblock In {\em Proceedings of the 2021 ACM-SIAM Symposium on Discrete
  Algorithms (SODA)\/} (2021), SIAM, pp.~522--539.

\bibitem{alstrup2019constructing}
{\sc Alstrup, S., Dahlgaard, S., Filtser, A., St{\"o}ckel, M., and
  Wulff-Nilsen, C.}
\newblock Constructing light spanners deterministically in near-linear time.
\newblock In {\em 27th Annual European Symposium on Algorithms (ESA 2019)\/}
  (2019), Schloss Dagstuhl-Leibniz-Zentrum fuer Informatik.

\bibitem{baswana2006faster}
{\sc Baswana, S., and Kavitha, T.}
\newblock Faster algorithms for approximate distance oracles and all-pairs
  small stretch paths.
\newblock In {\em 2006 47th Annual IEEE Symposium on Foundations of Computer
  Science (FOCS'06)\/} (2006), IEEE, pp.~591--602.

\bibitem{Baswana2010}
{\sc Baswana, S., Kavitha, T., Mehlhorn, K., and Pettie, S.}
\newblock Additive spanners and ($\alpha$, $\beta$)-spanners.
\newblock {\em {ACM} Transactions on Algorithms 7}, 1 (Nov. 2010), 1--26.

\bibitem{baswana2007simple}
{\sc Baswana, S., and Sen, S.}
\newblock A simple and linear time randomized algorithm for computing sparse
  spanners in weighted graphs.
\newblock {\em Random Structures \& Algorithms 30}, 4 (2007), 532--563.

\bibitem{chechik2013new}
{\sc Chechik, S.}
\newblock New additive spanners.
\newblock In {\em Proceedings of the twenty-fourth annual ACM-SIAM symposium on
  Discrete algorithms\/} (2013), SIAM, pp.~498--512.

\bibitem{cohen1998fast}
{\sc Cohen, E.}
\newblock Fast algorithms for constructing t-spanners and paths with stretch t.
\newblock {\em SIAM Journal on Computing 28}, 1 (1998), 210--236.

\bibitem{elkin2019almost}
{\sc Elkin, M., Gitlitz, Y., and Neiman, O.}
\newblock Almost shortest paths with near-additive error in weighted graphs.
\newblock {\em arXiv preprint arXiv:1907.11422\/} (2019).

\bibitem{elkin2020improved}
{\sc Elkin, M., Gitlitz, Y., and Neiman, O.}
\newblock Improved weighted additive spanners.
\newblock {\em arXiv preprint arXiv:2008.09877\/} (2020).

\bibitem{elkin2005approximating}
{\sc Elkin, M., and Peleg, D.}
\newblock Approximating k-spanner problems for k> 2.
\newblock {\em Theoretical Computer Science 337}, 1-3 (2005), 249--277.

\bibitem{knudsen2017additive}
{\sc Knudsen, M. B.~T.}
\newblock Additive spanners and distance oracles in quadratic time.
\newblock {\em arXiv preprint arXiv:1704.04473\/} (2017).

\bibitem{liestman1993additive}
{\sc Liestman, A.~L., and Shermer, T.~C.}
\newblock Additive graph spanners.
\newblock {\em Networks 23}, 4 (1993), 343--363.

\bibitem{peleg1989graph}
{\sc Peleg, D., and Sch{\"a}ffer, A.~A.}
\newblock Graph spanners.
\newblock {\em Journal of graph theory 13}, 1 (1989), 99--116.

\bibitem{peleg1987optimal}
{\sc Peleg, D., and Ullman, J.~D.}
\newblock An optimal synchronizer for the hypercube.
\newblock In {\em Proceedings of the sixth annual ACM Symposium on Principles
  of distributed computing\/} (1987), pp.~77--85.

\bibitem{roditty2004dynamic}
{\sc Roditty, L., and Zwick, U.}
\newblock On dynamic shortest paths problems.
\newblock In {\em European Symposium on Algorithms\/} (2004), Springer,
  pp.~580--591.

\bibitem{seidel1995all}
{\sc Seidel, R.}
\newblock On the all-pairs-shortest-path problem in unweighted undirected
  graphs.
\newblock {\em Journal of computer and system sciences 51}, 3 (1995), 400--403.

\bibitem{woodruff2010additive}
{\sc Woodruff, D.~P.}
\newblock Additive spanners in nearly quadratic time.
\newblock In {\em International Colloquium on Automata, Languages, and
  Programming\/} (2010), Springer, pp.~463--474.

\end{thebibliography}

\end{document}